\documentclass{amsart}

\usepackage[dvips]{graphicx}
\graphicspath{{./Fig/}}

\usepackage{subfigure,bbm}
\usepackage[latin1]{inputenc}
\usepackage{verbatim}

\newtheorem{proposition}{Proposition}
\newtheorem{theorem}{Theorem}

\newtheorem{definition}{Definition}

\def\N{{\mathbb N}}

\def\P{{\mathbb P}}
\def\E{{\mathbb E}}
\def\Q{{\mathbb Q}}
\def\k{\kappa_s}
\def\p{p_s}

\newcommand\ind[1]{\mathbbm{1}_{\{#1\}}}

\def\cal{\mathcal}

\def\Var{\mathrm{Var}}

\def\etal{{\em et al. }}
\author[Chabchoub]{Yousra Chabchoub}
\address[Y. Chabchoub, C. Fricker, Ph. Robert]{INRIA Rocquencourt,  RAP
    project, Domaine de Voluceau, 78153 Le Chesnay, France. }
\email{Yousra.Chabchoub@inria.fr}
\email{Christine.Fricker@inria.fr}
\email{Philippe.Robert@inria.fr}

\author[Fricker]{Christine Fricker}

\author[Guillemin]{Fabrice Guillemin}
\address{Orange Labs, F-22300 Lannion, France}
\email{Fabrice.Guillemin@orange-ftgroup.com}

\author[Robert]{Philippe  Robert}

\title[Statistical Characterization of Flows in Internet Traffic]{On the Statistical Characterization of Flows in Internet Traffic with Application to Sampling}

\date{\today}

\keywords{Flow statistics, Statistical Models, Pareto distribution, Poisson Approximation.}

\begin{document}
\begin{abstract}
A new  method of  estimating some statistical characteristics of TCP flows in the
Internet is developed  in this  
paper.  For this  purpose, a new set  of random variables (referred to  as observables) is
defined. When dealing with sampled traffic,  these observables can easily be computed from
sampled data.   By adopting a convenient  mouse/elephant dichotomy also  {\em dependent on
traffic},  it  is   shown  how  these  variables  give  a  reliable  statistical
representation  of  the number of packets transmitted by large  flows during successive time
intervals with an appropriate duration.   A mathematical  framework  is developed to
estimate the accuracy of the  method.  As an application, it is shown how one can estimate
the number of large  TCP flows when  only sampled traffic is  available.  The algorithm
proposed is  tested against experimental data collected  from different types of IP networks. 
\end{abstract}

\maketitle

\hrule

 \tableofcontents 

\vspace{-5mm}

\hrule

\vspace{5mm}

\section{Introduction}
In Internet traffic  a flow is  classically defined as the  set of those  packets with  the same
source and  destination IP addresses  together with the  same source and  destination port
numbers and the same protocol type.  It  is well known that if large TCP flows carry the
prevalent  part of traffic (in Bytes), most of flows are small (in number of packets).  A formal
definition of ``large''  and  ``small''  will be given later in the paper. As it will  be seen, it may depend
on the  context; in a first step, the discussion  is kept informal. 

We investigate in  this paper how to characterize the statistical  properties of the sizes
of large flows (notably their number of packets)  in Internet traffic. It is commonly observed
in the  technical literature and in  real experiments that  the total size (in  packets or
bytes)  of  such  flows has  a  heavy  tailed  distribution.  In practice,  however,  this
characterization  holds only  for very  large values  of  the flow  size. Consequently, in order  to
accurately estimate the tail of the  size probability distribution, a large number of large
flows is  necessary. To increase the  sample size when  empirically estimating probability
distribution tails, one is led to increase  the length of the observation period.  But the
counterpart is that the distribution of the flow  size can no more be described in terms of simple
probability distributions,  of the Pareto type  for example. This  is due to the  fact that
traffic is not stationary over long time  periods, for instance because of daily
variations of interactive services (video, web, etc.). 

Actually, numerous approaches  have been proposed in the technical  literature in order to
model large flows as well as their superposition properties. One can roughly classify them
in  two categories: signal  processing models  and statistical  models.  Using  ideas from
signal     processing,     Abry    and     Veitch~\cite{Abry},     see    also     Feldman
\etal~\cite{279346,300387}  and  Crovella   and  Bestravos~\cite{Crovella},  describe  the
spectral properties  of the time series associated  with IP traffic by  using wavelets. In
this way,  a characterization of long  range dependence (the Hurst  parameter for example)
can be provided. Straight lines in the  log-log plot of the power spectrum support some of
the ``fractal''  properties of the IP  traffic, even if they  may simply be  due to packet
bursts  in  data  flows. See  Rolland  \etal~\cite{Rolland}. 

Signal  processing  tools  provide
information on aggregated traffic but not on characteristics on individual TCP flows, like
the  number   of  packets  or  their   transmission  time.   For   statistical  models,  a
representation  with  Poisson  shot  noise  processes  (and  therefore  some  independence
properties)  has  been  used  to  describe  the  dynamics of  IP  traffic,  see  Hohn  and
Veitch~\cite{1133559},  Duffield \etal~\cite{Duffield2}, Gong \etal  \cite{Gong}, Barakat \etal    \cite{Barakat}   and    Krunz   and
Makowski~\cite{KM}   for    example.    In    Ben
Azzouna       \etal~\cite{ICC},       Loiseau \etal~\cite{Loiseau,Loiseau:02} and Gong
\etal\cite{Gong}, the distribution of the size of large flows  is represented by a Pareto
distribution, i.e. a probability  distribution whose tail decays  on a polynomial scale. \nocite{Mitzen},

The starting point of some of these analyses is  the need for understanding the relation between the distribution of the
number $\hat{S}$ of sampled packets when performing packet sampling and the distribution
of the flow size $S$. The problem can be described as follows:   $\P(\hat{S}=j)=Q(\P(S=\cdot), j)$, $j\geq 1$, with
\[
Q(\phi,i)=p^j\sum_{\ell=j}^{+\infty} \binom{\ell}{j}(1-p)^{\ell-j}\phi(\ell).
\]
The problem then  consists of finding a distribution $\phi_0$ maximizing some functional ${\cal L}(\phi)$
so that the relation $\P(\hat{S}=j)=Q(\phi, j)$ holds. See Loiseau
\etal~\cite{Loiseau:02}  for an extensive discussion of the current literature where our
algorithm is called  ``stochastic counting''. As it will be seen in the following, we will not rely on
the maximum likelihood ratio of distributions in our approach but on estimations of some
averages to estimate some key parameters. 

\medskip
\paragraph{\bf Statistical Characterization Method}
We  develop  in this  paper  an  alternative method  of obtaining  a   statistical
description  of  the size  of  large  flows  in IP  traffic  by  means of  a  Pareto
distribution: Statistics  are collected during  successive time windows of  limited length
(instead of one single  time window for the whole trace).  It  must be emphasized that this
characterization  in terms  of  a Pareto  distribution  does not  rely  on the  asymptotic
behavior of the tail  distribution but only on statistics on some  range of values for the
sizes of flows.  

The  advantage  of  the proposed  method  is  that  with  a  careful procedure,  a  simple
statistical characterization  is possible and seems to  be quite reliable as  shown by our
experiments  for various sets  of traffic  traces.  The  intuitive reason  for considering
short time periods is that on such  times scales, flows exhibit only one major statistical
mode (typically  a Pareto behavior).  In larger  time windows, different modes  due to the
wide variety of flows and non-stationarity in IP traffic necessarily appear.  (See Feldman
\etal~\cite{300387}.)   This  approach  allows  us  to establish  a  reliable  statistical
characterization of  flows which is used to  infer information from sampled  traffic as it
will be seen.  The counterpart of that the distribution of the {\em total} size of a large
flow (obtained when considering the complete traffic trace) cannot be obtained directly in
this way since the trace is cut into small pieces.

An algorithm is proposed to obtain  the statistical representation of large flows when all
the  packets of  the  trace  are available.   The  constants used  in  our algorithms  are
explicitly expressed as  either universal constants (independent of  traffic) or constants
depending on traffic : Length of  the observation window, definition of TCP flows referred
to as  large flows,  etc.  The procedure  invoked to  estimate flow statistics  should not
depend on some hidden pre-processing of  the trace. Our algorithms determine {\em on-line}
the constants depending on the traffic.  This  is, in our view, one important aspect which
is sometimes neglected in the technical literature

\medskip
\paragraph{\bf Application to Sampled Traffic}
The  basic motivation  for developing  a  flow characterization  method is  to infer  flow
characteristics from sampled  data.  This is notably the case  for sampling processes such
as the 1-out-of-$k$  sampling scheme implemented by CISCO's  NetFlow \cite{NetFlow}, which
greatly degrades information on flows. What we  advocate in this paper is that it is still
possible  to  infer   relevant  characteristics  on  flows  from   sampled  data  if  some
characteristics of the flow size can be  confidently described by means of a simple Pareto
distribution.   By using  the statistical  representation  described above,  we propose  a
method of inferring the number of large flows from sampled traffic.

The proposed method relies on a  new set of
random variables,  referred to  as observables and  computed in successive  time intervals
with fixed length.  Specifically, these random variables count the number of flows sampled
once,  twice or  more  in the  successive  observation windows.  The  properties of  these
variables can be obtained through  simple characteristics, in particular {\em mean values}
of variables  instead of {\em remote quantiles}  of the tail distribution,  which are much
more difficult  to accurately estimate.   By developing a convenient  mathematical setting
(Poisson approximation methods),  it is moreover possible to  show that quantities related
to  the observables  under consideration  are close  to Poisson  random variables  with an
explicit bound on the error.  This Poisson approximation is the key result to estimate the
total number of large flows.  


\medskip
\paragraph{\bf Organization of the paper}
The organization of the paper is as  follows. A statistical description of large TCP flows
is  presented in  Section~\ref{experiments}, this  representation is  tested  against five
exhaustive sets  of traffic  traces: three from  the France Telecom  (FT) commercial  IP network
carrying residential ADSL  traffic and  two others from Abilene  network.  An algorithm is
developed in this section to compute the characteristics  of the Pareto distributions
describing flows. In Section~\ref{sampledtraffic}, some assumptions on sampled traffic are
introduced and the observables for describing traffic are defined. The  mathematical 
properties are analyzed in light  of Poisson approximation methods in
Section~\ref{properties}.  The results developed in  this section  are  crucial to  infer
the statistics  of an  IP  traffic from  sampled data. Experiments with  the five sets of
sampled traces used in this  paper are presented and  discussed   in  Section~\ref{ssec}.
Some   concluding  remarks  are   presented  in  Section~\ref{conclusion}.

\section{Statistical Properties of Flows}\label{experiments}
This section  is devoted to  a statistical study  of the size (the number  of packets) of
flows in a  limited time window  of duration  $\Delta$.  The  goal of  this section  is
show  that  a simple statistical representation  of the flow size  can be
obtained for various sets of traffic traces.

\subsection{Assumptions and Experimental Conditions}

\subsubsection*{The sets of traces used for testing theoretical results}
For the experiments carried out in the following sections, several sets of traces will be considered: Commercial IP traffic, namely ADSL
traces from the France Telecom (FT) IP collect network,  and traffic issued from campus networks (Abilene III traces). Their
characteristics are given in Table~\ref{traces}.

\begin{table*}[hbtp]
\caption{Characteristics of traffic traces considered in experiments.\label{traces}}
\begin{center}
\begin{tabular}{lrcl}\hline
Name& Nb. IP packets & Nb. TCP Flows&Duration\\ \hline
ADSL Trace A& 271 455 718& 20 949 331 & 2 hours\\ 
ADSL Trace B Upstream &   54 396 226&2 648 193&2 hours \\ 
ADSL Trace B Downstream &53 391 874 &2 107 379 &2 hours\\ 
Abilene III Trace A &62 875 146 & 1 654 410 &  8 minutes\\ 
Abilene III Trace B & 47 706 252 &1 826 380 &8 minutes \\ \hline
 \end{tabular}
\end{center}
\end{table*}

The Abilene traces 20040601-193121-1.gz (trace A) and 20040601-194000-0.gz (trace B) can
be found at the url  {http://pma.nlanr.net/Traces/Traces/long/ipls/3/}.

\subsubsection*{Time Windows}
Traffic will  be observed in successive  time windows with length  $\Delta$.  In practice,
the  quantity $\Delta$  can vary  from a  few seconds  to several  minutes  depending upon
traffic characteristics  on the link  considered. 

The ideal value of $\Delta$ actually depends on the targeted application. For  the design
of network elements considering the flow level (e.g., flow aware routers, measurement
devices, etc.), it is necessary to estimate the requirements in terms of memory to store
the different flow descriptors. In this context, $\Delta$ may be of the order of few
seconds. The same order of magnitude is also adapted to  anomaly detection,  for instance
for detecting  a sudden  increase in  the  number of flows.  For the  computation of
traffic  matrices, $\Delta$  can be  several minutes  long (typically 15 minutes). In our
study, the ``adequate'' values for $\Delta$ are of the order of several seconds. See the
discussion below. 

\subsubsection*{Mice and Elephants}

With regard to the analysis of the  composition of traffic, in light of earlier studies on
IP traffic  (see Estan and  Varghese~\cite{859719}, Papagiannaki \etal~\cite{Taft}  or Ben
Azzouna \etal~\cite{Nadia}), two  types of flows are identified:  small flows with few
packets (referred to as  mice) and the other flows will be referred  to as  elephants.  In
commercial IP  traffic, this simple  traffic  decomposition  can be justified  by  the
predominance  of web  browsing  and peer-to-peer traffic  giving rise to  either signaling
and  very small file transfers  (mice) or else file downloads (elephants).

This  dichotomy may  be  more delicate  to  verify in  a different  context  than the  one
considered in Ben Azzouna \etal~\cite{Nadia}.  For  LAN traffic, for example, there may be
very large  amounts of  data transferred at  very high speed.  As it  will be seen  in the
various IP traces used in our analysis,  the distinction between mice and elephants has to
be handled with  care and in our case  is dependent on  the type of
traffic  considered.  The distinction  between the  constants depending  on the  trace and
``universal'' constants is, in our view,  a crucial issue. It amounts to precisely stating
which constants are {\em depend} on traffic.  This aspect is generally (unduly in our opinion)
neglected in  traffic measurement studies.  In  particular, the variable  $\Delta$ and the
dichotomy mice/elephants are dependent on the trace, as explained in the next section.

\subsection{Heavy Tails}
\label{experimentsbis}
The fact that the distribution of the size $S$ of a large TCP flow is heavy tailed is well
known.  Experiments and theoretical  results on  the superposition  of ON-OFF  heavy tailed
traffic  have  justified  the  self  similar  nature  of  IP  traffic,  see  Crovella  and
Bestravos~\cite{Crovella}.  Although the heavy tailed  property of the size of large flows
is commonly admitted, little attention has been paid to identify properly a class of heavy
tailed  distributions  so  that the  corresponding  parameters  can  be estimated  for  an
arbitrary traffic trace with a significant duration.

One of the  reasons for this situation  is that the most common  heavy tailed distributions
$G(x)=\P(S\geq  x)$  (e.g., Pareto, i.e., $G(x)=C/x^\alpha$ for $x\geq  b$ and some $\alpha >0$,  or Weibull, i.e.,  $G(x)=\exp(-\nu
x^\beta)$ for some $\beta >0$ and $\nu>0$)  have a  very  small number  of parameters  and consequently  a
limited of  number of possible  degrees of  freedom for describing the  distribution of the  sizes of
flows.  For  this reason, such  a distribution can  rarely  represent the statistics  of the
total number of packets transmitted by a flow in a trace of arbitrary duration.

As  a matter  of  fact, if  a  traffic trace  is sufficiently  long,  some non  stationary
phenomena may  arise and the  diversity of file  sizes may not be  captured by one  or two
parameters. For example, with a Pareto distribution, the function $x\to G(x)$ in a log-log
scale should be  a straight line. The statistics  of the file sizes in the  traces used in
our experiments are depicted in Figure~\ref{fig1} and~\ref{fig11} for an ADSL traffic
trace from the France Telecom backbone IP collect network and for a traffic trace from
Abilene network, respectively. 

\begin{figure}[hbtp]
\begin{center}
\rotatebox{-90}{\scalebox{0.25}{\includegraphics{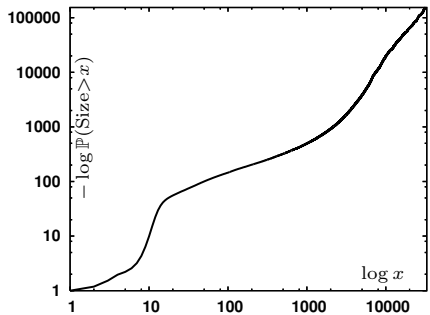}}}
\put(-180,0){\rotatebox{90}{\put(-80,-42){$\scriptstyle-\log \P\left(\text{Size}>x\right)$}}}
\put(-35,-110){$\scriptstyle \log x$}
\end{center}
\caption{Statistics of  the number of packets $S$ of a flow  for  ADSL A (2 hours):  the quantity  $-\log(\P(S>x))$ as a function of $\log(x)$. \label{fig1}}
\end{figure}

\begin{figure}[hbtp]
\begin{center}
\rotatebox{-90}{\scalebox{0.25}{\includegraphics{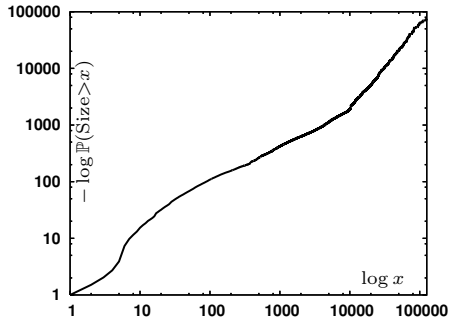}}}
\put(-180,0){\rotatebox{90}{\put(-80,-42){$\scriptstyle-\log \P\left(\text{Size}>x\right)$}}}
\put(-35,-110){$\scriptstyle \log x$}
\caption{Statistics of  the number of packets  $S$ of a flow  for ABILENE A trace (8 minutes):  the quantity
  $-\log(\P(S>x))$ as a function of $\log(x)$. \label{fig11}}
\end{center}
\end{figure}

Figure~\ref{fig1} and~\ref{fig11} clearly show that for the two traffic traces considered,
the file size exhibits a multimodal behavior: At least {\em several} straight lines should
be  necessary to properly describe   these  distributions. These  figures also  exhibit the
(intuitive) fact  that has been noticed in  earlier experiments: The longer  the trace is,
the more  marked is the multimodal  phenomenon.  (See Ben Azzouna  \etal~\cite{ICC} for a
discussion.)

The {\em key observation} when characterizing a traffic trace is the fact that if the duration $\Delta$ of the successive 
time intervals used for computing traffic parameters  is appropriately chosen, then the distribution of
the size  of the main  contributing flows in  the time interval  can be represented  by a
Pareto distribution.  More precisely, there exist
$\Delta$,  $B_{min}$, $B_{max}$  and  $a>0$ such  that if  $S$  is the  number of  packets
transmitted by a flow in $\Delta$ time units, then $\P(S\geq x \mid S\geq B_{min})\sim
P_\alpha(x)$  for $B_{min}\leq x\leq B_{max}$ with 
\begin{equation}\label{pareto}
P_\alpha(x)\stackrel{\text{def.}}{=}\left(\frac{B_{min}}{x}\right)^a, \text{ for } x\geq B_{min},
\end{equation}
and furthermore the proportion of large  flows with  size greater than $B_{max}$ is less than $5\%$.
The parameter $B_{min}$ is usually referred to as the location parameter and $a$ as the
shape parameter.    

In other words,  if the time interval  is sufficiently small then the  distribution of the
number of packets  transmitted by a large flow  has one dominant Pareto mode and  therefore can
confidently be characterized by a unique Pareto distribution.  The  algorithm used  to  validate this  result is  described in Table~\ref{algo}.  It is  run from  the beginning  of  the trace;  in practice  a couple  of minutes  is
sufficient to obtain  results for the constants $\Delta$, $B_{min}$, $B_{max}$.  The algorithm
is of course valid  when the total trace is available for at  least an interval of several
minutes.  In  the case  of sampled traffic for which this  algorithm cannot be  used, another
method will be proposed in Section~\ref{sampledtraffic}.

\begin{table}[hbtp]
\caption{Algorithm for Identifying $\Delta$ and the Pareto Distribution.\label{algo}}
\hrule

\begin{itemize}
\item[---] { $\Delta$  is fixed so that at least  $1000$ flows  have more than  $20$ packets.} 
\item[---] { $B_{max}$ is defined  as the smallest integer such that  less than $5\%$ of the flows  have a size greater than $B_{max}$.}
\item[---]A Least  Square Method, see Deuflhard and Hohmann~\cite{Num} for example, is
  performed  to get a linear interpolation  in a log-log   scale  of   the  distribution
  of sizes   between  $B_{min}$  and   $B_{max}$.  The constant  $B_{min}$ is  chosen  as
  the  smallest integer  such that  the $L_2$-distance in the sense of least  square method
  with  the approximating straight line   is less than  $2.10^{-3}$. The slope of  the
  line   gives the value of the parameter $a$.  
\end{itemize}

\hrule

\end{table}

The  quantity $B_{min}$  defines the  boundary  between mice  and {\em  elephants} in  the
trace. A {\em mouse} is a flow with  a number of packets less than $B_{min}$.  An elephant
is a  flow such that its  number of packets during  a time interval of  length $\Delta$ is
greater  than or equal  to $B_{min}$.   By definition  of $B_{max}$,  flows whose  size is
greater than $B_{max}$ represent a small fraction of the elephants.

\subsection{Experiments with Synthetic and Real Traffic Traces}
Some   experiments  have  been   done  using   artificial  traces   with  a   real  Pareto
distribution. For  these traces,  the algorithm described in Table~\ref{algo}  has been used  without any
modification: A  time window is defined  when at least  $1000$ flows of size  greater than
$20$ packets are detected.   As it can be seen, the identification  of the exponent $a$ is
quite good. Note that, because only  Pareto distributed flows are present the minimal size
$B_{min}$ of elephants is smaller than in real traffic.

\begin{figure}[hbtp]
\begin{center}
\subfigure[Pareto  $a=1.85$. Estimation: $\hat{a}=1.84$,
  $B_{min}=9$, $B_{max}=100$]{
\rotatebox{-90}{\scalebox{0.25}{\includegraphics{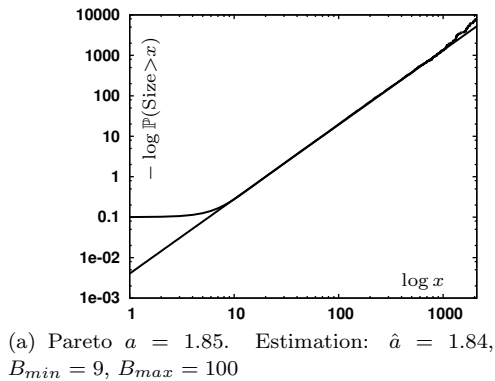}}}
\put(-170,10){\rotatebox{90}{\put(-80,-42){$\scriptstyle-\log \P\left(\text{Size}>x\right)$}}}
\put(-35,-110){$\scriptstyle\log x$}
}

\subfigure[Pareto  $a=2.5$. Estimation: $\hat{a}=2.48$,
  $B_{min}=11$, $B_{max}=65$]{
\rotatebox{-90}{\scalebox{0.25}{\includegraphics{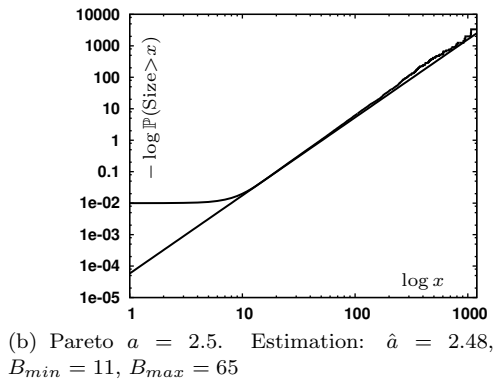}}}
\put(-170,10){\rotatebox{90}{\put(-80,-42){$\scriptstyle-\log \P\left(\text{Size}>x\right)$}}}
\put(-35,-110){$\scriptstyle\log x$}
}
\end{center}
\caption{Synthetic traces with $10^6$ flows with a Pareto distribution}\label{synth}
\end{figure}

Experimental results with real traces, for the ADSL A and Abilene A traffic traces, are displayed in
Figures~\ref{fig2a} and \ref{fig2b}, respectively.  The same
algorithm has been run for the ADSL trace B Upstream and Downstream as well as for the
Abilene III B trace. The benefit of the algorithm is that the distribution of the number of packets in
elephants can always be represented by a unimodal Pareto distribution if the duration of
$\Delta$ is adequately chosen by using the algorithm given in Table~\ref{algo}. Results are summarized in Table~\ref{tab2}. 

\begin{table}[hbtp]
\caption{Statistics of the elephants for the different traffic traces.\label{tab2}}
\begin{tabular}{llllll}\hline
&{\tiny ADSL A}& {\tiny ADSL B Up} & {\tiny ADSL B Down }&{\tiny Abilene A} &{\tiny Abilene  B}  \\ \hline
$\Delta$ (sec)& 5& 15&15&2&2\\
$B_{min}$&20 & 29&39&89&79\\ 
$B_{max}$&94& 154&128&324&312 \\ 
$a$&1.85&1.97&1.50&1.30&1.28\\ 
\hline
\end{tabular}
 \end{table}

\begin{figure}[hbtp]
\begin{center}
\subfigure[ ADSL A trace --  $\Delta = 5$s]{
{{\rotatebox{-90}{\scalebox{0.25}{\includegraphics{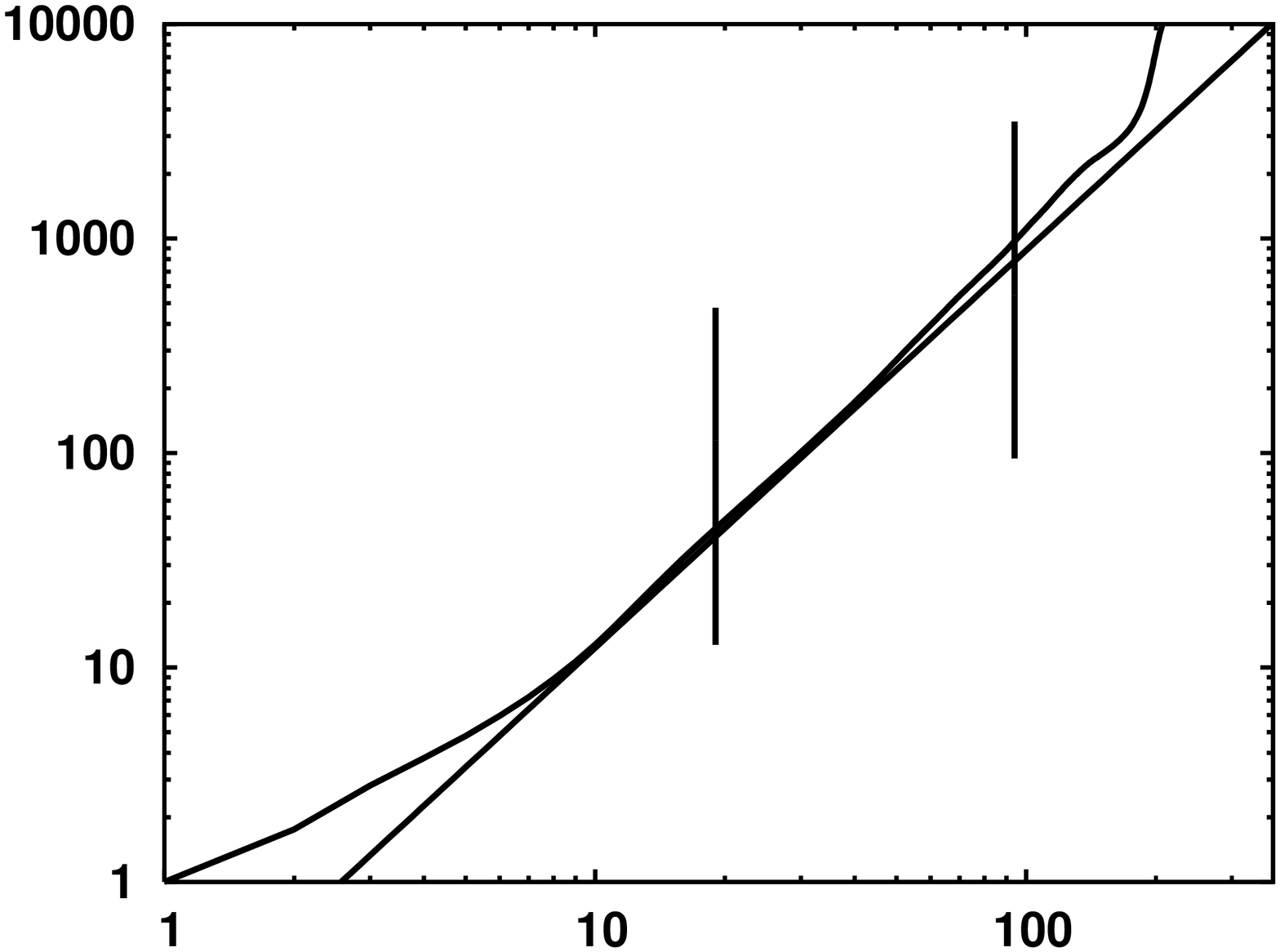}}}}
\put(-180,0){\rotatebox{90}{\put(-80,-42){$\scriptstyle-\log \P\left(\text{Size}>x\right)$}}}
\put(-35,-110){$\scriptstyle \log x$}
\put(-50,-80){$ B_{max}$}
\put(-90,-30){$ B_{min}$}}
}
\subfigure[ADSL B Down trace -- $\Delta = 15$ seconds]{
{{\rotatebox{-90}{\scalebox{0.25}{\includegraphics{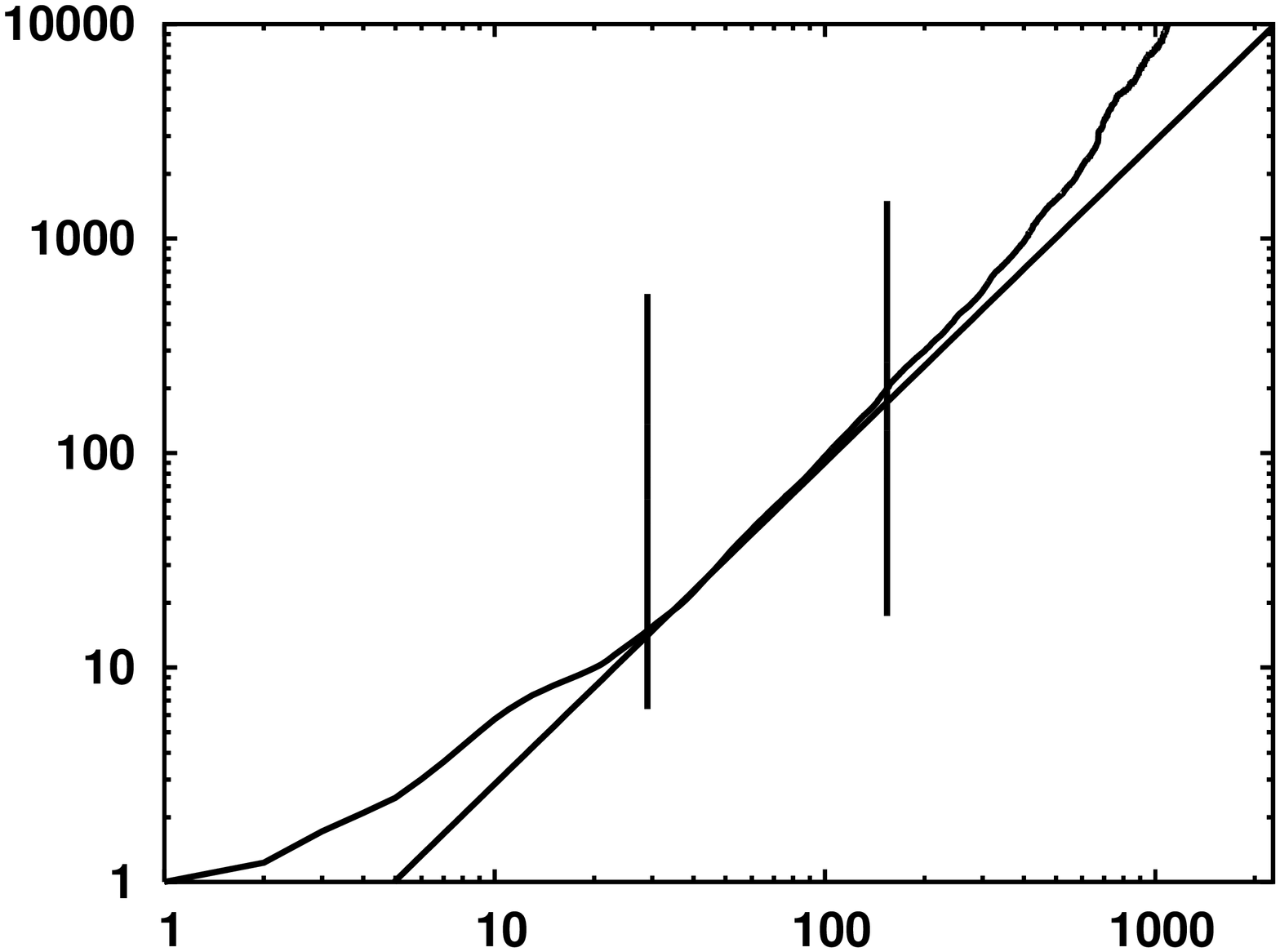}}}}
\put(-180,0){\rotatebox{90}{\put(-80,-42){$\scriptstyle-\log \P\left(\text{Size}>x\right)$}}}
\put(-35,-110){$\scriptstyle\log x$}
\put(-50,-80){$ B_{max}$}
\put(-90,-30){$ B_{min}$}}
}

\subfigure[ADSL B Up trace -- $\Delta =15$ s]{
{{\rotatebox{-90}{\scalebox{0.25}{\includegraphics{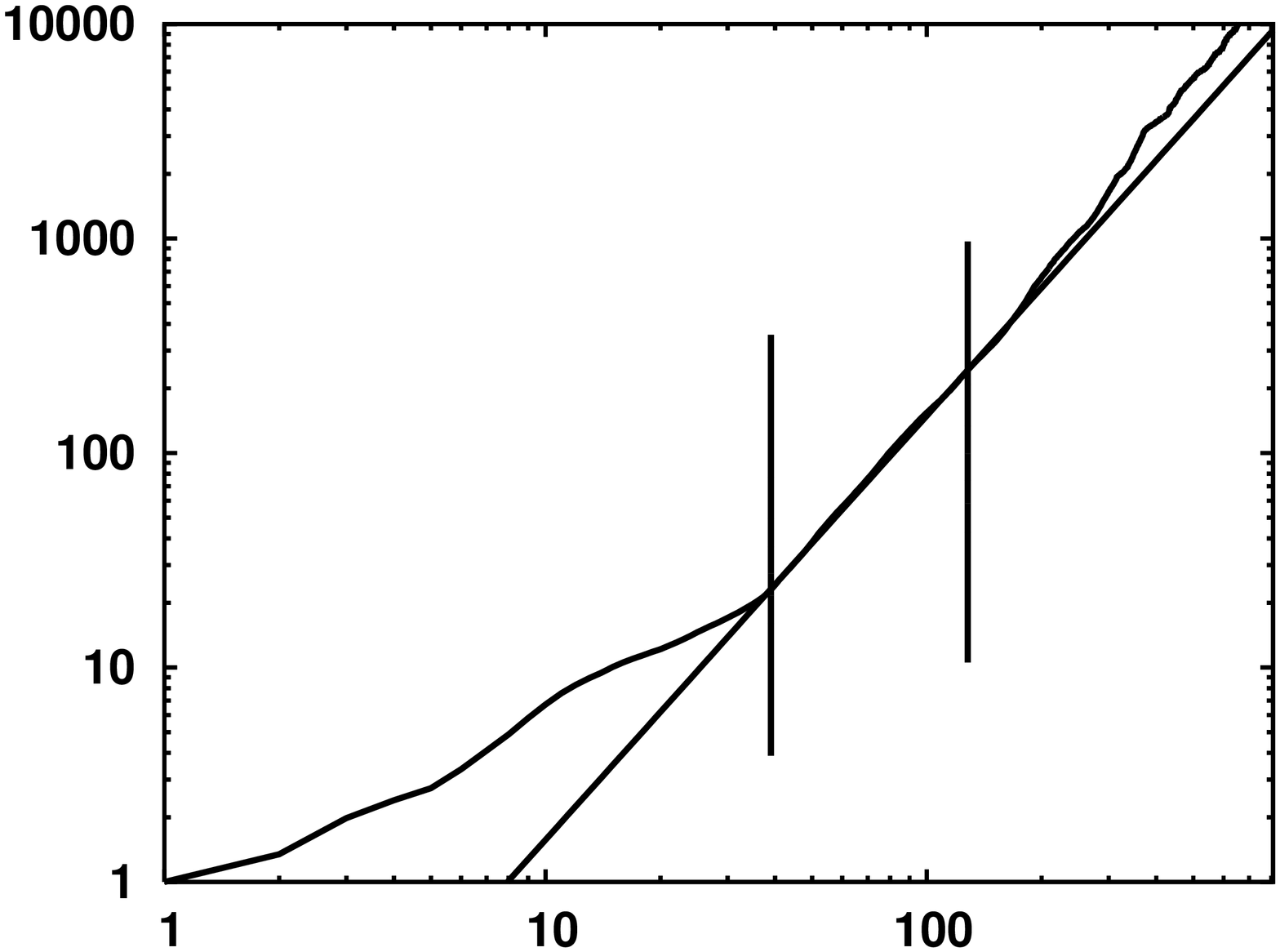}}}}
\put(-180,0){\rotatebox{90}{\put(-80,-42){$\scriptstyle-\log \P\left(\text{Size}>x\right)$}}}
\put(-35,-110){$\scriptstyle\log x$}
\put(-50,-80){$ B_{max}$}
\put(-90,-30){$ B_{min}$}}
}
\end{center}
\caption{Statistics of  the flow size (number of packets) in a time interval of length $\Delta=15$\label{fig2a}}
\end{figure}

\begin{figure}[hbtp]
\begin{center}
\subfigure[Abilene A trace -- $\Delta = 2$ s]{
{{\rotatebox{-90}{\scalebox{0.25}{\includegraphics{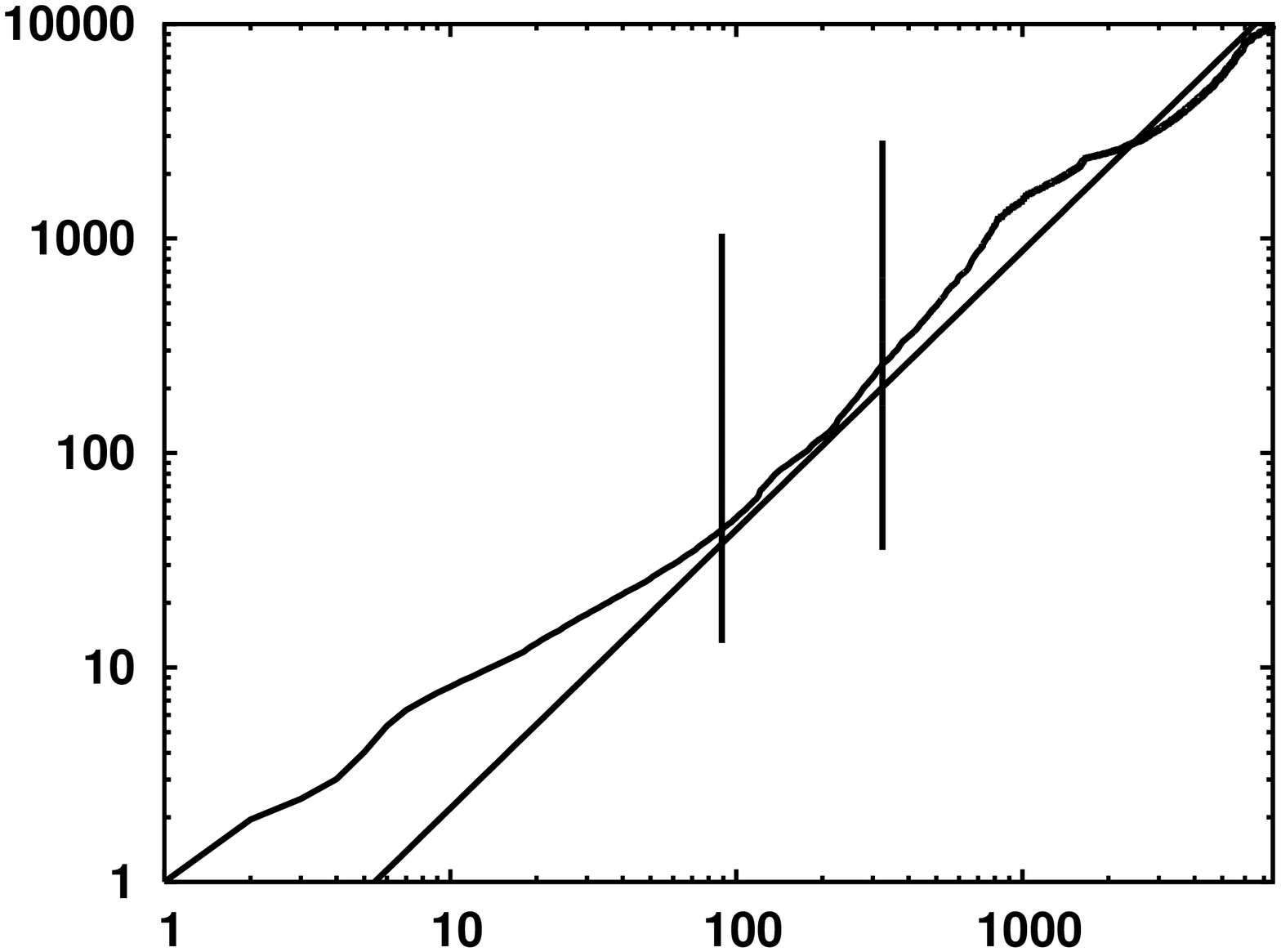}}}}
\put(-180,0){\rotatebox{90}{\put(-80,-42){$\scriptstyle-\log \P\left(\text{Size}>x\right)$}}}
\put(-35,-110){$\scriptstyle\log x$}
\put(-50,-80){$ B_{max}$}
\put(-90,-30){$ B_{min}$}}
}
\subfigure[Abilene B trace -- $\Delta = 2$ s]{
{{\rotatebox{-90}{\scalebox{0.25}{\includegraphics{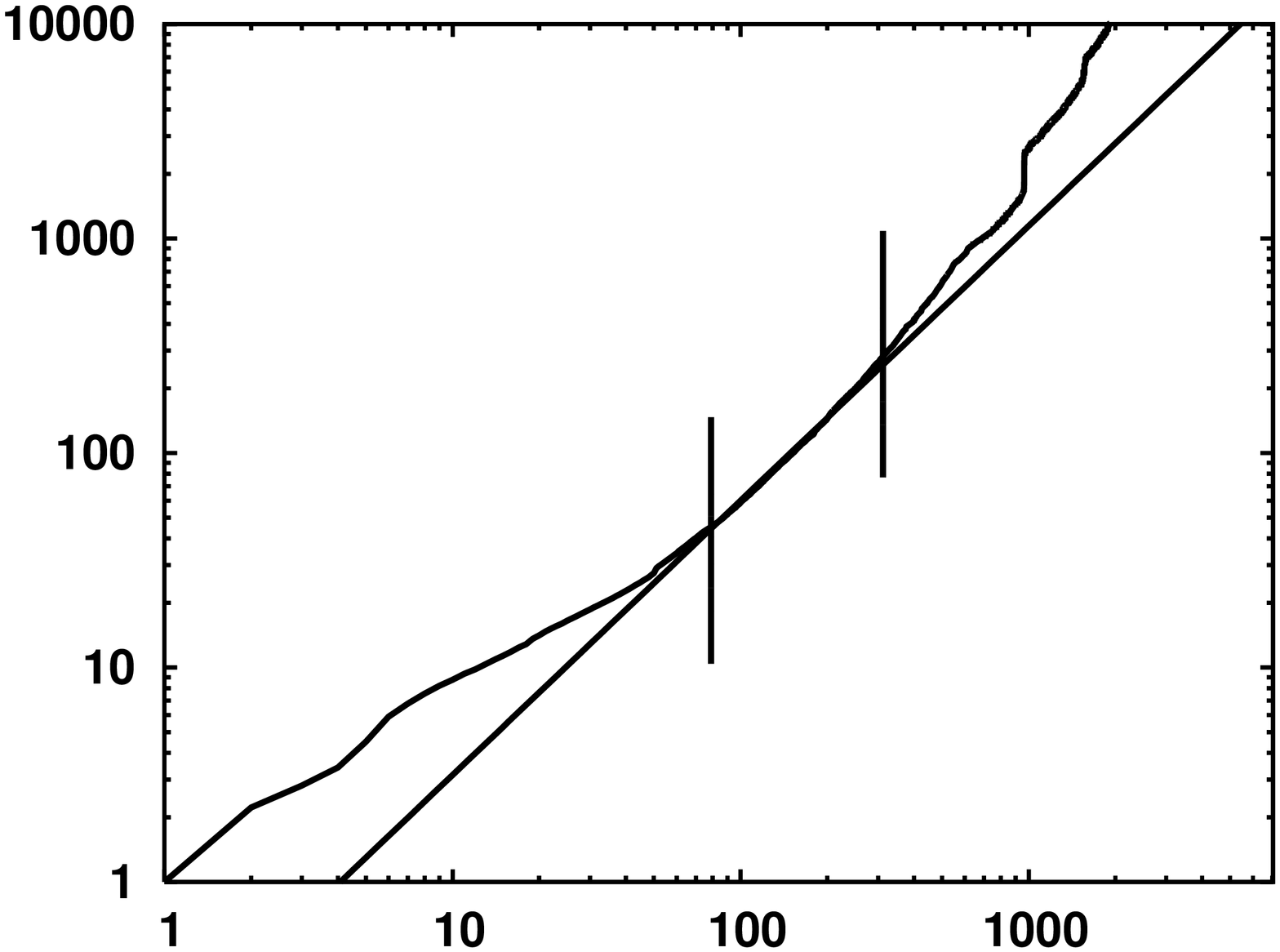}}}}
\put(-180,0){\rotatebox{90}{\put(-80,-42){$\scriptstyle-\log \P\left(\text{Size}>x\right)$}}}
\put(-35,-110){$\scriptstyle\log x$}
\put(-50,-80){$ B_{max}$}
\put(-90,-30){$ B_{min}$}}
}
\end{center}
\caption{Statistics of  the flow size  (number of packets)  in a time interval of length $\Delta$ for the Abilene traces.\label{fig2b}}
\end{figure}

\subsection{On the choice of parameters} 
We discuss in this section the various parameters used by the algorithm.

\subsubsection*{Fixed parameters and parameters depending on traffic}
There are four basic parameters for the model which are determined by the trace:
$\Delta$ (duration of time window for statistics), the range of values $[B_{min},B_{max}]$ for the Pareto
distribution and the exponent $a$ of this distribution.  These parameters are discussed
below.   

Additionally there are ``universal'' (i.e. independent of the trace): the minimal number
of flows to make statistics, set to $1000$ here,  the proportion, $5\%$, of flows of size
$\geq B_{max}$,  and the level of accuracy, $2.10^{-3}$ here, of the least  square method
to determine $B_{min}$ and $B_{max}$.  

\subsubsection*{Parameter $B_{min}$}
It  turns out  that for  commercial (ADSL)  traffic, the  value of  $B_{min}$ is  close to
$20$. This  value is fairly  common in earlier  studies for classifying ADSL  traffic.  It
should be noted that this value is not at all universal since, in our view, it does depend
on traffic. The examples  with  Abilene traces,  see below, which  contain significantly
bigger elephants, shows  that the corresponding values should be  higher than $20$ (around
$80$ in our  example). 

The two types of traffic are  intrinsically different: ADSL traffic
is mainly composed  of peer to peer traffic (with  a huge number of small  flows and a few
file transfers  of limited size because of  the segmentation of large  files into chunks),
while  Abilene traffic comprises  large file  transfers issued  from campus  networks.  In
order to maximize the range for  the Pareto description, the variable $B_{min}$ is defined
as the smallest value for which the linear representation (in the log scale) holds.

\subsubsection*{Parameter $\Delta$}
This parameter $\Delta$ is determined in a simple way by our algorithm. 
According to the various  experiments, the parameter $\Delta$ can  be taken in some
range of values where the Pareto  representation still holds.  On the one hand, $\Delta$
has  to be taken large enough so that  sufficiently many  packets arrive  in time
intervals of  duration $\Delta$ to derive  reliable  estimations of  the Pareto distribution.
An experiment with ADSL  A trace with $\Delta=1s$ gives only $63$  flows of size more than
$20$ which  is not  enough  to obtain reliable  statistics.  A ``correct'' value  in this
case is $5s$. Experiments show that higher values (like $10s$) do not change significantly
the Pareto property observed in this case.

On the other hand, $\Delta$ should not  be too large so that the statistical properties (a
Pareto  distribution in our  case) can  be identified,  i.e., so  that the  statistics are
unimodal.  See  Figures~\ref{fig1}   and~\ref{fig11}  which  illustrate  situations  where
statistics are done on the complete trace,  i.e. when $\Delta$ is taken equal to the total
duration of the trace. In these examples,  the piecewise linear aspect of the curves suggests,
for both cases, there is at least a bi-modal Pareto behavior.

\subsection{Discussion}
As it will be seen in the following,  the above statistical model gives interesting results
to extract information from sampled traffic. It has nevertheless some shortcomings which
are now discussed. 

\subsubsection*{A partial information when $\Delta$ is small}.
It should be noted that the parameters computed in a time window of length $\Delta$ do not
give  a complete  description  of the  distribution  of the  size of  a  large  flow,  since
statistics  are done  over a  limited  time horizon.  The procedure  provides therefore  a
fragmented information.  

To obtain a complete description of the statistics of the size of
flows, it  would be necessary to relate the statistics from successive  time windows of
length $\Delta$.  We do  not know how to do that yet. Nevertheless, as  it will be seen in
the following, this fragmented information can be recovered from  sampled traffic and
it will be  used to give a good estimation  on the number of active large  flows at a given
time.  This  incomplete but useful description of  the statistics is, in  some sense, the
price to pay to have a simple  estimation of the statistics of flows.

\subsubsection*{An incomplete description of large  flows in a time window of size $\Delta$}
The representation with a Pareto distribution is for elephants (with size greater than $B_{min}$) 
whose size is less than $B_{max}$. In particular, it does not give  any information on the
statistics of flows with size greater than $B_{max}$. But note that, by definition, less that
$5\%$ of the total number of flows have a size greater than $ B_{max}$. This is however 
a source of errors when, as in Section~\ref{properties}, the Pareto representation
is used on the interval $[B_{min},+\infty]$ instead of $[B_{min},B_{max}]$

\section{Sampled Traffic: Assumptions and Definition of Observables}
\label{sampledtraffic}

In the previous section, an algorithm to describe the distribution of large flows by means
of a  unimodal distribution has  been introduced.   Now, it is  shown how to  exploit this
algorithm in the context of packet sampling  in the Internet. Packet sampling is a crucial
issue  when performing  traffic  measurements  in  high speed  backbone
networks. As a  matter of fact, a  fundamental problem related to the  computation of flow
statistics from  traffic crossing very high speed  transmission links is that,  due to the
enormous number of packets handled by routers, only a reduced amount of information can be
available to the network operator. 

Packet  sampling is in this context an efficient method
of reducing  the volume of data to  analyze when performing measurements  in the Internet.
One popular  technique consists  of picking up  one packet  every other $\k$  packets with
$\k=100$,  $500$,  $1000$   in  practice.   (This  sampling  scheme   is  referred  to  as
1-out-of-$\k$ packet  sampling in the  technical literature.)  This method  is implemented
for instance in  CISCO routers, namely NetFlow facility~\cite{NetFlow}  widely deployed in
operational networks today. It suffers  from different shortcomings well identified in the
technical literature, see for instance Estan \etal~\cite{Estan}.

We describe in this section the different  assumptions made on traffic in order to develop
an  analytical evaluation of  our method  of inferring  flow statistics.   Throughout this
paper, high speed transmission links (at least 1~Gbit/s) will be considered.

\subsection{Mixing condition} When observing traffic, packets are assumed to be sufficiently interleaved
so that those packets of a same flow  are not back-to-back but mixed with packets of other
flows.  This  introduces some  randomness  in the  selection  of  packets when  performing
sampling. In particular, when $K$ flows are active in a given time window and if  the
$i$th flow comprises $v_i$ packets during that period, then the probability of selecting a 
packet of the $i$th flow is assumed to be equal $v_i/(v_1+v_2+\cdots+v_K)$. This property will
be referred  to as  \emph{mixing condition} in  the following  and is formally  defined as
follows. A variant of this property is, implicitly at least, assumed in the existing
literature. See, e.g. Duffield \etal~\cite{Duffield} and Chabchoub \etal~\cite{Chabchoub:01}.

\begin{definition}[Mixing Condition]\label{defmix}
If $K$ TCP flows  are active during a time interval of  duration $\Delta$, traffic is said
to be  mixing if for  all $i$, $1\leq  i\leq K$, the  total number $\hat{v}_i$  of packets
sampled  from the  $i$th flow  during that time interval  has the  same distribution  as  the analog
variable in the following scenario: at each sampling instant a packet of the $i$th flow is
chosen with probability $v_i/V$ where $v_i$ is the number of packets of the $i$th flow and
$V=v_1+\cdots+v_K$.
\end{definition}

This amounts to claim  that with regard to sampling, the probability of selecting a packet of a
given flow is proportional to the total number of packets of this flow. 

One alternative  would consist of assuming that  the probability of selecting  a packet of
the $i$th  flow is  $1/K$, the  inverse of the  total number  of flows.   This assumption,
however,  does not take  into account  the respective  contributions  of the
different flows to the total volume and  thus may be inaccurate.   If all $K$ flows had
the same distribution with a  small variance, then this assumption  would not much differ from
the mixing condition. Note however that the variance of Pareto distributions can be
infinite if the shape parameter $a$ is less than $2$.  Hence, this leads us 
to suppose that the mixing condition holds  and that the probability of selecting a packet
from flow $i$ is indeed $v_i/V$.

\subsection{Negligibility assumption}

We consider traffic on  very high speed links and it then  seems reasonable to assume that
no flows contribute a significant proportion of global traffic. In other words, we suppose
that the contribution of  a given flow to global traffic is  negligible. In the following,
we go  one step further by assuming  that in any time  window, the number of  packets of a
given flow is negligible  when compared to the total number of  packets in the observation
window. By using the notation of the previous section, this amounts to assuming  that for any
flow $i$, the  number of packet $v_i$ is  much less than $V$. Furthermore,  we even impose
that the squared value of $v_i$ is much less than $V$. We specifically formulate the above
assumptions as follows.

\begin{definition}[Negligibility condition]\label{defeps}
In any window of length $\Delta$, the square of the number of packets of every flow is
negligible when compared to the total number of packets $V$ in the observation
window. There specifically exists some $0< \varepsilon \ll 1$ such that for all
$i=1,\ldots,K$, $v_i^2/V \leq \varepsilon$. 
\end{definition}
The above assumption implies that no flows are dominating when observing traffic on a high
speed transmission  link. Table~\ref{tabpois} shows that  this is the case  for the traces
used in  our experiments.  There is thus  no bias  in the sampling  process, which  may be
caused by the  fact that some flows are oversampled because  they contribute a significant
part of traffic. This assumption is  reasonable for commercial ADSL traffic because access
links are often the  bottlenecks in the network. For instance, ADSL  users may have access
rates of a few  Mbit/s, which are negligible when compared against  backbone  links of 1 to 10
Gbit/s.  Moreover,  the bit  rate achievable by  an individual  flow rarely exceeds  a few
hundreds of  Kbit/s.  In the case of  transit networks carrying campus  traffic, the above
assumption may be  more questionable since bulk data transfers may  take place in Ethernet
local area networks and individual flows may achieve bit rates of several Mbit/s.

\begin{table*}[hbtp]
\caption{The quantity $\E(v_1^2)/\E(V)$ for traffic traces considered in experiments.}\label{tabpois}
\begin{tabular}{|l|l|l|l|}\hline
Trace& $\Delta=5$sec & $\Delta=10$sec & $\Delta=15$sec\\ \hline
ADSL A & 0.000146&0.000159 &0.000168 \\ \hline
ADSL B up&0.001100 &--- &0.001335 \\ \hline
ADSL B Down& 0.002199&0.002543 &0.002732 \\ \hline
\end{tabular} 

\medskip

\begin{tabular}{|l|l|l|l|l|}\hline
Trace& $\Delta=1$sec & $\Delta=2$sec & $\Delta=3$sec& $\Delta=5$sec\\ \hline
Abilene A & 0.055001&0.068833 &0.064813&0.072768\\ \hline
\end{tabular} 

\medskip

\begin{tabular}{|l|l|l|}\hline
Trace& $\Delta=1$sec & $\Delta=2$sec \\ \hline
Abilene B &0.011786 &0.013804 \\ \hline
\end{tabular} 

\end{table*}

\subsection{The Observables}
We now introduce the different variables used to infer flow characteristics. These variables
are based  only upon sampled  data; they can  be evaluated when analyzing  NetFlow records
sent by  routers of an  IP network. For  this reason, these  variables are referred  to as
observables.   Because of  packet sampling,  recall that  the original  characteristics of
flows (for instance their duration or their original number of packets) cannot be directly
observed.

The observables  considered in  this paper  to infer flow  characteristics are  the random
variables $W_j$, $j  \geq 1$, where $W_j$ is  the number of flows sampled $j$  times during a
time interval of  duration $\Delta$.  The averages of the random variables $W_j$ are
in fact the key quantities used to infer the characteristics of flows from sampled data.

The random variables $W_j$, $j \geq 1$ are formally defined as follows: Consider a time interval of length
$\Delta$ and  let $K$ be the  total number of large flows  present in this  time interval.  Each
flow $i\in\{1,\ldots,K\}$ is  composed of $v_i$ packets in  this time interval. Let denote by $\hat{v}_i$
the number of times that  flow $i$ is sampled. The random variable $W_j$ is
simply defined by
\begin{equation}
\label{defWj}
W_j=\ind{\hat{v}_1=j}+\ind{\hat{v}_2=j}+\cdots+\ind{\hat{v}_K=j}.
\end{equation}

In  practice, if  $\Delta$ is  not too  large,  the data  structures used  to compute  the
variables $W_j$  are reasonably simple.   Moreover, as it  will be seen in  the following,
provided that  $\Delta$ is appropriately chosen,  the statistics of the  number of packets
transmitted by elephants  during successive time windows with  duration $\Delta$ are quite
robust.  Consequently, the variables $W_j$ inherit  also this property. When the number
of large  flows is  large, the  estimation of the  asymptotics of  their averages  from the
sampled traffic is easy in practice.  Theoretical results on these variables are derived in
the next section.

\section{Mathematical Properties of the Observables} \label{properties}

\subsection{Definitions and Le Cam's inequality}

For $j\geq 0$, the variable $W_j$  defined by Equation~\eqref{defWj} is a sum of Bernoulli
random  variables, namely
$$ W_j=\ind{\hat{v}_1=j}+\ind{\hat{v}_2=j}+\cdots+\ind{\hat{v}_K=j},  $$
where $\hat{v}_i$ is the  number of times that the $i$th flow  has been sampled.  If these
indicator functions  were independent,  by assuming that  $K$ is  large, one could  use to
estimate  the distribution  of  $W_j$ either  via a  Poisson approximation  (in  a rare  event
setting) or via a central limit theorem (in  a law of large numbers context).  Since the total
number of samples is known, the sum of the random variables $\hat{v}_i$ for $i=1,\ldots,K$
is known and then, the Bernoulli variables defining $W_j$ are {\em not} independent.

To overcome this  problem, we make use of  general results on the sum  of Bernoulli random
variables.  Let   us  consider   a  sequence  $(I_i)$   of  Bernoulli   random  variables,
i.e.  $I_i\in\{0,1\}$.  The  distance  in  total variation  between  the  distribution  of
$X=I_1+\cdots+I_i+\cdots$ and a Poisson  distribution with parameter $\delta>0$ is defined
by
\begin{multline*}
\|\P(X\in \cdot)-\P(Q_{\delta}\in\cdot)\|_{tv} \stackrel{\text{def.}}=   \sup_{A\subset \N} \
\left|\P(X\in A)-\P(Q_{\delta}\in A)\right| \\
= \frac{1}{2}\sum_{n\geq 0} \left|\P(X=n)-\frac{\delta^n}{n!}e^{-\delta}\right|.
\end{multline*}
The Poisson distribution $Q_{\delta}$   with mean  $\delta$ is such that 
\[
\P(Q_{\delta}=n)=\frac{\delta^n}{n!} \exp(-\delta).
\]
\noindent
Note that the total variation distance is a strong distance since it is uniform with
respect to all events, i.e.,  for all subset s $A$ of $\N$,
$$|\P(X\in A)-\P(Q_{\delta}\in A)|\leq \|\P(X\in \cdot)-\P(Q_{\delta}\in\cdot)\|_{tv}.$$

The  following result (see Barbour \etal~\cite{Barbour}) gives  a tight  bound on  the total  variation distance  between the
distribution of  $X$ and the  Poisson distribution with  the same expected value  when the
Bernoulli variables are independent. In spite of the fact that this  result  is not directly applicable in our case, we shall show in the following how to use it to obtain information on the distributions of the observables $W_j$.

\begin{theorem}[Le Cam's Inequality] 
If the random variables $(I_i)$ are independent and if $X = \sum_i I_i$, then
\begin{multline}\label{lecam}
\|\P(X\in \cdot)-\P(Q_{\E(X)}\in\cdot)\|_{tv}\leq \sum_{i} \P(I_i=1)^2 = \E(X)-\Var(X)
\end{multline}
\end{theorem}
If $X$ is a Poisson distribution then $\Var(X)=\E(X)$, the above relation shows that to
prove the convergence to a Poisson distribution one has only to prove that the expectation
of the random variable is arbitrarily close to its variance.

\subsection{Estimation of the mean value of the observables}
\label{calculmean}

We  consider the  $1$-out-of-$\k$ deterministic  sampling technique,  where one  packet is
selected  every other $\k$  packets.  In  addition, we  suppose that  traffic on  the  observed link
   is   sufficiently    mixed   so   that   the   mixing    condition   given   by
Definition~\ref{defmix} holds  and that there are  no dominating flows in  traffic so that
the negligibility condition (Definition~\ref{defeps}) also pertains.

It is assumed that during a time interval of length $\Delta$, there are $K$ flows composed
of at  least $B_{min}$ packets,  where $B_{min}$ is defined  in Section~\ref{experiments}.
It has been seen  that the number of packets in these  flows follows a Pareto distribution
defined  by Relation~\eqref{pareto}  for some  exponent $a$  and parameters  $B_{min}$ and
$B_{max}$. Let $S$  be a random variable  whose distribution is given by
Relation~\eqref{pareto} for all $x\geq B_{min}$. From our experiments,  $S$ is the size of
a ``typical'' flow whose size is in the interval $[B_{min},B_{max}]$. See the discussion
at the end of Section~\ref{experiments} for the flows of size greater than $B_{max}$.  Of course
the sizes of mice are not represented by this random variable. The variable $V$ denotes
the total  number of packets in  the 
observation window, note that it includes not only the elephants but also the mice.

Note  that $V$ is the  sum of the number of  packets in elephants and  mice. If $v_i$ is
the number of packet  in the $i$th elephant,   then    $v_i$   has   the    same   Pareto
distribution   as    $S$   (i.e., $v_i\stackrel{\text{dist.}}{=}S$)  and  $V \geq
v_1+v_2+\cdots+v_K$.   The  difference $V  - v_1-v_2-\cdots-v_K$ is the number of packets of mice.

\begin{proposition}[Mean Value of the Observables]\label{boundobs}
If $K$ elephants are active in  a time window of length $\Delta$, the mean number $\E(W_j)$ of
flows sampled $j$ times,  $j\geq 1$,  satisfies the relation 
\begin{align}\label{asymp}
\left|\frac{\E(W_j)}{K}-\Q_{j}\right|\leq \p\E\left(\frac{S^2}{V} \right),
\end{align}
where $\Q$  is the probability distribution defined by 
\begin{align*}
\P(\Q = j) \stackrel{def}{=} \Q_{j}=\E\left(\frac{{\left(\p S\right)}^j}{j!}e^{-{\p}S}\right),
\end{align*}
and $\p=1/\k$ is  the sampling rate.
\end{proposition}
From Equation~\eqref{asymp} one gets that the larger the total volume $V$ of packets is,
the better is the approximation of $\E(W_j)/K$ by $\Q_j$. 
\begin{proof}
The number of times $\hat{v}_i$ that the  $i$th flow is sampled in the time interval is  given by
\[
\hat{v}_i=  B_{1}^i+B_{2}^i+\cdots+B_{\p V}^i,
\]
where, due to the mixing condition, $B_\ell^i$ is equal to one if the  $\ell$th sampled
packet is from the $i$th flow, which event occurs with probability $v_i/V$. Note that the total
number of sampled packets is $ \p V $.  

Conditionally on the values of the set ${\cal F}=\{v_1, \ldots, v_K\}$, the variables
$(B_\ell^i, \ell \geq 1)$ are independent Bernoulli variables. For $1\leq i\leq K$, Le
Cam's Inequality~\eqref{lecam} gives therefore the 
relation 
\[
\left\|\P(\hat{v}_i\in\cdot\mid {\cal F})-Q_{\p {v}_i}\right\|_{tv}
\leq \p\frac{v_i^2}{V}.
\]
By integrating with respect to the variables $v_1, \ldots, v_K$,  this gives the relation
\[
\left\|\P(\hat{v}_i\in\cdot)-\Q\right\|_{tv} \leq \p\E\left(\frac{v_i^2}{V}\right).
\]
In particular, for $j\in\N$, $\left|\P(\hat{v}_i=j)-\Q_j\right| \leq \p\E\left(S^2/{V}\right)$. Since
\[
\E(W_j)=\sum_{i=1}^K\P(\hat{v}_i=j),
\]
by summing on $i=1, \ldots, K$,  one gets
\[
\left|\E(W_j)-K\Q_j\right|\leq 
\p K\E\left(\frac{S^2}{V}\right)
\]
and the result follows.
\end{proof}
If the number of packets per flow were  constant, then $\Q$ would be a Poisson distribution
with  parameter $\p  S$,  the  variable $S$  being  in this  case  a  constant. The  above
inequality shows  that at the first  order the expected  value of $W_j$ is  $\p\E(S)$. The
expression of $\Q$, however, indicates that higher order moments of $S$ play a significant
role. For  example, if  the variable $S$  has a  significant variance, then  the classical
rough reduction, which consists of assuming that the size of a sampled elephant is $\p S$,
is no longer valid for estimating the original size of the elephant.

Under the negligibility condition, we deduce that 
$$
\left|\frac{\E(W_j)}{K}-\Q_{j}\right|\leq \p\varepsilon,
$$ where $\varepsilon$ appears in Definition~\ref{defeps} and is assumed to much less than
1. This  implies that Inequality~\eqref{asymp} is  tight and the  quantity $\E(W_j)/K$ can
accurately be approximated by the quantity $\mathbb{Q}_j$, when no flows are dominating in
traffic.

We are now ready to state the main result needed for estimating  the number $K$ of
elephants from sampled data.  
\begin{proposition}[Asymptotic Mean Values]
Under  the same assumptions as those of  Proposition~\ref{boundobs}, 
\begin{equation}\label{equiv}
\lim_{K\to +\infty} \frac{\E(W_{j+1})}{\E(W_{j})}\sim 1-\frac{a+1}{j+1}
\end{equation}
and
\begin{equation}\label{Kel}
\lim_{K\to +\infty} \frac{\E(W_j)}{K}\sim a(\p B_{min})^{a} \frac{\Gamma(j-a)}{j!},
\end{equation}
if $B_{max}>>1$ and $\p B_{min}<<1$, where $\Gamma$ is the classical Gamma function defined by
\[
\Gamma(x)=\int_{0}^{+\infty} u^{x-1} e^{-u}\,du, \quad x> 0.
\]
\end{proposition}

\begin{proof}
For $j\geq 1$,
\begin{equation*}
\Q_{j}=\E\left(\frac{{\left(\p S\right)}^j}{j!}e^{-{\p}S}\right) 
\sim aB_{min}^a \frac{\p^{a+1}}{j!}\int_{B_{min}}^{+\infty} (\p u)^{j-a-1} e^{-\p u}\,du
\end{equation*}
and then 
\begin{equation*}
\Q_{j}\sim aB_{min}^a \frac{\p^{a} }{j!} \int_{\p B_{min}}^{+\infty} u^{j-a-1} e^{-u}\,du 
\sim a(\p B_{min})^{a} \frac{\Gamma(j-a)}{j!},
\end{equation*}
since $\p B_{min}\sim 0$. Therefore, by using the relation $\Gamma(x+1)=x\Gamma(x)$ we obtain  the equivalence
\[
\frac{\Q_{j+1}}{\Q_j}\sim \frac{j-a}{j+1}.
\]
The proposition follows by using the fact that the upper bound of Equation~\eqref{asymp} of
Proposition~\ref{boundobs} goes to $0$ by the law of large numbers. 
\end{proof}

As it will be seen later in the next section, Relation~\eqref{equiv} is used to estimate the exponent $a$ of the Pareto distribution of the number of packets of elephants, the quantities $\E(W_j)$ and $\E(W_{j+1})$ being  easily derived from sampled traffic.  The quantity  $K$ will be estimated  from Relation~\eqref{Kel}. The estimation of the parameter $B_{min}$
from sampled traffic as well as the correct choice of the integer $j$ will be discussed  in
the next section.

\section{Applications}\label{ssec}
\subsection{Traffic parameter inference algorithm}
In  this section,  it  is assumed  that only  sampled  traffic is  available. The  methods
described in  Section~\ref{experiments} to infer  the statistical properties of  the flows
cannot  be applied and another  algorithm  has  to be  defined. For  the
experiments carried  out in the  present section, the  sampling factor $\p=1/\k$  has been
taken  equal to  $1/100$.   To infer  flow characteristics,  we  have to  give the  proper
definition of the  mouse and elephant dichotomy (the parameter  $B_{min}$) and to estimate
the  coefficient  of   the  corresponding  Pareto  distribution  (the   parameter  $a$  in
Relation~\eqref{pareto}).

Relation~\eqref{equiv} gives the following equivalence, for $j\geq 1$ sufficiently large
so that the impact of mice on $\E(W_j)$ is negligible, 
\begin{equation}\label{est1}
a\sim a(j)\stackrel{\text{def.}}{=}(j+1)\left(1-\frac{\E(W_{j+1})}{\E(W_j)}\right)-1,
\end{equation}
and Relation~\eqref{Kel} yields  an estimate of the number of elephants, i.e. the number of
flows with a  number of packets greater than or equal to  $B_{min}$; we specifically have
\begin{equation}\label{est2}
K\sim K(j)\stackrel{\text{def.}}{=}\frac{j!\,\E(W_j)}{a(j)(\p B_{min})^{a(j)} \Gamma(j-a(j))}.
\end{equation}
These estimations greatly  depend on some of the key parameters used to obtain a convenient and
confident  Pareto representation of the size of the flows, in particular the size of the time
window $\Delta$ and the lower bound $B_{min}$ for the elephants. The variable $\Delta$ is chosen so that 
\begin{enumerate}
\item the number of flows sampled twice is sufficiently large in order to obtain a
significant number of  samples so that the estimation of the mean values of the random
variables $W_j$ for $j \geq 2$ is accurate; this requires that $\Delta$ should not be
 too small, 
\item  $\Delta$ is not too large in order to preserve  the unimodal Pareto representation (see Section~\ref{experiments} for a discussion).
\end{enumerate} 
To count the average number of  flows sampled $j$ times,  the parameter $j$  should be 
chosen as large as possible in order to neglect the impact of mice (for which the Pareto representation does not
hold) but not too large so that the statistics are robust to compute the mean value
$\E(W_j)$. 

In the experimental work reported below, special attention has been paid  to the choice of the {\em
  universal} constants, i.e., those  constants used in the analysis of sampled data, that do not depend on the traffic trace considered. 
In our opinion, this is a crucial in  an accurate  inference of traffic parameters from sampled data. These
  constants are defined in the algorithm given in Table~\ref{algoo}. 

\begin{table}[hbtp]
\caption{Algorithm used to identify $\Delta$ and the Pareto parameter from sampled traffic.\label{algoo}}
\hrule 

\smallskip

\begin{itemize}
\item[---] 
Choose $\Delta$  so that $80 \leq \E[W_2] \leq 100$;
\item[---] Choose $j$ so that $|a(j)-a(j+1)|$ computed with Equation~\eqref{est1}  is
  minimized with for all $j$ such that $\E[W_j] \geq 5$.
\item[---] $B_{min}$  is the smallest integer so that the probability that a  flow of size
  greater than $B_{min}$ is sampled more than $j$  times is greater than $\p/10$;
\end{itemize}
\hrule
\end{table}
\medskip

\subsection{Experimental results}

Concerning the estimation of the constants $B_{min}$, the numerical results obtained by
using the algorithm given in Table~\ref{algoo}  are
presented in Table~\ref{tab3}, where the values of the different $B_{min}$ estimated by
the algorithm are compared against  the values given in Section~\ref{experiments}. As it
can be observed, the proposed algorithm  yields  a rather  conservative definition of
elephants (i.e.,  flows of size greater than or equal to $B_{min}$). 

\begin{table}[htbp]
\caption{Elephants for the France Telecom ADSL and the Abilene traffic traces.\label{tab3}}
\begin{tabular}{|l|l|l|l|l|l|}\hline
&{\tiny ADSL A}& {\tiny ADSL B Up} & {\tiny ADSL B Down }&{\tiny Abilene A} &{\tiny Abilene  B}  \\ \hline
 $B_{min}$&20 & 29&39&89&79\\ \hline
estimated $B_{min}$&21 & 45&45&77&77\\ \hline
\end{tabular}
 \end{table}

The main  results are gathered  in Table~\ref{Results} giving  the quantities $K$  and $a$
estimated by  using Equations~\eqref{est1}  and \eqref{est2} for  different values  of the
parameters $j$.  These values are compared  against the experimental  values $a_{exp}$ and
$K_{exp}$, referred  to as  the ``real'' $a$  and $K$  obtained from the  complete traffic
traces in Section~\ref{experiments}.   The accuracy of the estimation  of $K$ is generally
quite good except for the Abilene A  trace where the error is significant although not out
of  bound.   A look  at  the corresponding  figure  in  Section~\ref{experiments} gives  a
plausible explanation for  this discrepancy: For this trace,  the Pareto representation is
not very precise.

Finally, it is worth noting from Table~\ref{Results} that the
estimation of the important parameter $a$ describing the statistics of flows   is also
quite accurate. The error in this table is defined as 
\[
\frac{K(j)-K_{exp}}{K_{exp}}.
\]

\begin{table*}[hbtp]
\caption{Estimations of the Number of Elephants from Sampled traffic}\label{Results}
\begin{center}
\begin{tabular}{|l|l|l|l|l|l|l|l|l|l|}\hline
Trace   &$\Delta$       &$j$ &${\scriptstyle \E(W_{j}})$&${\scriptstyle \E(W_{j{+}1})}$
&$a_{exp}$ &$a(j)$ &$K_{exp}$  &$K(j)$ &Error \\\hline
ADSL A        &5s      &3      &12.89 &3.33 &1.85    &1.95    &943.71 &1031.04&9.25\% \\\hline
ADSL B Do &15s     &4     &9.7  &4.75 &1.49   &1.55    &414.90  &404.13&2.59\% \\\hline
ADSL B Up    &15s     &4     &7.46  &2.97 &1.97    &2.00       &453.01  &462.68&2.13\% \\\hline
ABILENE A     &1s      &5     &6.04  &3.21 &1.38    &1.81    &217.44 &270.79 &24.53\% \\\hline
ABILENE B     &1s      &5    &6.1  &3.7 &1.36    &1.51    &209.12  &197.12&5.74\% \\\hline
\end{tabular}

\end{center}
\end{table*}

\noindent
{\bf Remark.} As pointed out  by Loiseau \etal~\cite{Loiseau:02}, the determination of $\Delta$ is crucial. 
Recall it is determined explicitly by the first step of our algorithm, see Table~\ref{algoo}.

\section{Conclusion}
\label{conclusion}

We have developed in this paper one method of characterizing flows in IP traffic by a few
parameters and  another one of inferring these parameters from sampled data obtained
via deterministic 1-out-of-$k$ sampling. For this purpose, we have made some restrictive
assumptions, which are in our opinion essential in order to establish an accurate
characterization of flows. The basic principle we have adopted consists of describing
flows in successive observation windows of limited length, which has to satisfy two
contradicting requirements. On the one hand, observation windows shall not to be too large
in order to preserve a description of flow statistics  as simple as possible, for instance
their size by means of a simple Pareto distribution.  
 
On the other hand, a sufficiently large number of packets has to be present in each
observation window in order to be able of computing flow characteristics with sufficient
accuracy, in particular  the tail of the distribution of the flow size. By assuming that
large flows (elephants) have a size which is Pareto distributed, we have developed an
algorithm to determine the optimal observation window length together with the parameters
of the Pareto distribution. The location parameter $B_{min}$ (see Equation~\eqref{pareto})
leads to a natural division  of the total flow population into two sets: those flows with
at least $B_{min}$ packets, referred to as elephants, and those flows with less than
$B_{min}$ packets,called mice. This method of characterizing flows has been tested against
traffic traces from the France Telecom and Abilene networks carrying completely different
types of traffic. 

For interpreting sampled data, we have made assumptions on the sampling process. We have
specifically supposed that flows are sufficiently interleaved in order to introduce some
randomness in the packet selection process (mixing condition) and that there are no
dominating flows so that there is no bias with regard to the  probability of sampling  a
flow (negligibility condition). These two assumptions allows us to establish rigorous
results for the number of times an elephant is sampled, in particular for the mean values
of the random variables $W_j$, $j \geq 1$.  

Of course, when analyzing sampled data, the original flow statistics are not known. In
particular, the length of the observation window necessary to characterize the flow size
by means of a unique Pareto distribution is unknown. To overcome this problem, we have
proposed an algorithm to fix the observation window length and the minimal length of
elephants. Then, by choosing the index $j$ sufficiently large so as to neglect the impact of
mice, the  theoretical results are used to complete the flow parameter inference. This
method has been  tested against the  Abilene and the France Telecom traffic traces and
yields satisfactory results. 

\providecommand{\bysame}{\leavevmode\hbox to3em{\hrulefill}\thinspace}
\providecommand{\MR}{\relax\ifhmode\unskip\space\fi MR }
\providecommand{\MRhref}[2]{%
  \href{http://www.ams.org/mathscinet-getitem?mr=#1}{#2}
}
\providecommand{\href}[2]{#2}

\end{document}